\documentclass[letterpaper,12pt]{article}
\usepackage{enumerate} 
\usepackage[margin=1in]{geometry}  
\usepackage{comment}
\usepackage{listings}
\usepackage{amsmath}
\usepackage{amsthm}
\usepackage{tikz}
\usepackage{caption}
\usepackage{array}
\usepackage{mdwmath}
\usepackage{multirow}
\usepackage{mdwtab}
\usepackage{eqparbox}
\usepackage{multicol}
\usepackage{amsfonts}
\usepackage{tikz}
\usepackage{multirow,bigstrut,threeparttable}
\usepackage{amsthm}
\usepackage{array}
\usepackage{bbm}
\usepackage{subfigure}
\usepackage{epstopdf}
\usepackage{mdwmath}
\usepackage{mdwtab}
\usepackage{eqparbox}
\usepackage{tikz}
\usepackage{latexsym}
\usepackage{cite}
\usepackage{amssymb}
\usepackage{bm}
\usepackage{amssymb}
\usepackage{graphicx}
\usepackage{mathrsfs}
\usepackage{epsfig}
\usepackage{psfrag}
\usepackage{setspace}
\usepackage[
            CJKbookmarks=true,
            bookmarksnumbered=true,
            bookmarksopen=true,
            colorlinks=true,
            citecolor=red,
            linkcolor=blue,
            anchorcolor=red,
            urlcolor=blue
            ]{hyperref}
\usepackage{algorithm}
\usepackage{algpseudocode}
\usepackage{stfloats}
\usepackage{algorithm}

\theoremstyle{plain}
\newtheorem{theorem}{Theorem}
\newtheorem{example}{Example}
\newtheorem{lemma}{Lemma}
\newtheorem{remark}{Remark}
\newtheorem{corollary}{Corollary}
\newtheorem{definition}{Definition}

\def \bE {\mathbb{E}}
\def \bR {\mathbb{R}}

\def \cN {\mathcal{N}}

\usepackage{xspace}

\newcommand{\Prob}{\mathbb{P}}

\definecolor{myblue}{rgb}{.8, .8, 1}
\definecolor{mathblue}{rgb}{0.2472, 0.24, 0.6} 
\definecolor{mathred}{rgb}{0.6, 0.24, 0.442893}
\definecolor{mathyellow}{rgb}{0.6, 0.547014, 0.24}

\newcommand{\ampmc}{\hat{a}_{\text{MPMC}}}
\newcommand{\asve}{\hat{a}_{\text{SVE}}}
\newcommand{\rmpmc}{\hat{r}_{\text{MPMC}}}
\newcommand{\rsve}{\hat{r}_{\text{SVE}}}

\usepackage{cleveref}
\crefname{lemma}{Lemma}{Lemmas}
\Crefname{lemma}{Lemma}{Lemmas}
\crefname{thm}{Theorem}{Theorems}
\Crefname{thm}{Theorem}{Theorems}
\crefformat{equation}{(#2#1#3)}

\DeclareMathOperator*{\argmax}{arg\,max}

\begin{document}

\title{A Multi-armed Bandit MCMC, with applications in sampling from doubly intractable posterior}
\author{Guanyang Wang\thanks{Guanyang Wang is with the Department of Mathematics, Stanford University. Email: \{guanyang\}@stanford.edu.}}

\maketitle

%
%
%
%
%
%
\begin{abstract}
    Markov chain Monte Carlo (MCMC) algorithms are widely used to sample from complicated distributions, especially to sample from the posterior distribution in Bayesian inference. However, MCMC is not directly applicable when facing the doubly intractable problem. In this paper, we discussed and compared two existing solutions -- Pseudo-marginal Monte Carlo and Exchange Algorithm. This paper also proposes a novel algorithm: Multi-armed Bandit MCMC (MABMC), which chooses between two (or more) randomized acceptance ratios in each step. MABMC could be applied directly to incorporate Pseudo-marginal Monte Carlo and Exchange algorithm, with higher average acceptance probability.  
\end{abstract}
\section{Introduction and Problem Formulation}\label{sec:introduction}
Sampling from the posterior is the central part in Bayesian inference. Suppose there is a  family of densities $p_\theta(x)$ on the sample space $x\in \mathcal{X}$, and a prior $\pi(\theta) $ on the parameter space $\Theta$. It is of our interest to sample from the posterior $\pi(\theta|x) \propto \pi(\theta) p_\theta(x)$. Assuming the prior and likelihood can be evaluated at every point, then Markov chain Monte Carlo (MCMC) would be the natural choice. The standard Metropolis-Hastings (M-H) Algorithm (Algorithm \ref{alg:M-H}) constructs a Markov chain on $\Theta$, with stationary distribution $\pi(\theta|x)$. 

\begin{algorithm}
\caption{Metropolis-Hastings Algorithm}\label{alg:M-H}
\hspace*{\algorithmicindent} \textbf{Input:} initial setting $\theta$, number of iterations $T$ \\

\begin{algorithmic}[1]

\For{$t= 1,\cdots T$}
\State Propose $\theta'\sim q(\theta'|\theta)$
\State Compute $$a = \frac{\pi(\theta'|x)q(\theta|\theta')}{\pi(\theta|x)q(\theta'|\theta)} = \frac{\pi(\theta')p_{\theta'}(x)q(\theta|\theta')}{\pi(\theta)p_{\theta}(x)q(\theta'|\theta)}$$
\State Draw $r \sim \text{Uniform}[0,1]$
\State \textbf{If} $(r< a)$ \textbf{then} set $\theta =\theta'$
\EndFor
\end{algorithmic}
\end{algorithm}

Here $a$ is often named as `acceptance ratio', and  $\min\{a,1\}$ is often called `acceptance probability'.

In real situations, however, the likelihood may often be intractable or computationally expensive.  In this scenario the likelihood function is known up to a normalizing constant, that is:
\[
p_\theta(x) = \frac{f_\theta(x)}{Z(\theta)},
\]
where $f_\theta(x)$ can be evaluated at every $(x, \theta)\in \mathcal{X}\times \Theta$, but $Z(\theta)$ is unknown. This intractable constants arise in many statistical problems and interesting models, such as image analysis \cite{besag1986statistical}, Gaussian graphical models \cite{roverato2002hyper}, Ising models \cite{pettitt2003efficient}.

Suppose one is still interested in sampling from the posterior $\pi(\theta|x)$, a standard Metropolis-Hastings algorithm (Algorithm \ref{alg:M-H}) gives  acceptance ratio:

\[
a = \frac{\pi(\theta')f_{\theta'}(x)q(\theta|\theta')}{\pi(\theta)f_{\theta}(x)q(\theta'|\theta)}\frac{Z(\theta)}{Z(\theta')},
\]
which cannot be calculated due to unknown ratio $\frac{Z(\theta)}{Z(\theta')}$.

This problem is known as `doubly intractable problem', as $Z(\theta)$ and $Z(\theta')$ are both unknown. Based on the idea of estimating the normalizing constant or the likelihood function, a wide range of techniques are proposed, such as  maximum pseudo-likelihood estimator \cite{besag1974spatial}, ratio importance sampling \cite{chen1997monte}, bridge sampling \cite{meng1996simulating}, path sampling \cite{gelman1998simulating}. These methods use different approaches to estimate $Z(\theta)$ and $Z(\theta')$, plugging the estimator into the expression of $a$. However, it breaks the detailed balance and causes the Markov chain not converging to the correct stationary distribution, which may cause problems.  

The Pseudo-marginal Monte Carlo is first introduced in \cite{beaumont2003estimation}. M{\o}ller et al. \cite{moller2006efficient}  proposed a `single auxiliary variable method' , which is a special case of the pseudo marginal Monte Carlo approaches. This algorithm is asymptotically exact, i.e., the Markov chain converges to the correct stationary distribution.  The convergence rate of Pseudo-marginal Markov chain can be found in \cite{andrieu2009pseudo}, \cite{andrieu2015convergence}.

In 2006, Murray et al. \cite{murray2012mcmc} proposed the `exchange algorithm', which provides a generalization of  M{\o}ller et al. \cite{moller2006efficient}. The exchange algorithm also proposes auxiliary variables in each iteration and is asymptotically exact. However, the exchange algorithm requires perfect sampling from $p_{\theta'}$ which is not practical in many cases. Usually sampling from $p_{\theta'}$ also requires doing MCMC and would be very slow.  Therefore there are several results working on proposing variants of exchange algorithm to tackle this problem. For example, Faming Liang et al. proposed `doubly Metropolis-Hastings sampler'\cite{liang2010double} and `adaptive exchange algorithm'\cite{liang2016adaptive}, Alquier et al. proposed   `Noisy Monte Carlo' \cite{alquier2016noisy}. All these algorithms run Markov chains using approximate transition kernels, thus the stationary distribution is still no longer the exact posterior distribution and may not even exist.

The modified pseudo-marginal Monte Carlo which we will introduce later and the exchange algorithm construct Markov chains with randomized acceptance probability by introducing auxiliary variables at each iteration, which will both be referred to as `Randomized MCMC' (RMCMC) hereafter. 

This paper provides a comparison of the two algorithms and a new MCMC algorithm, which chooses between these two algorithms adaptively at each iteration, obtaining better acceptance probabilities. In Section \ref{review}, we reviewed the two RMCMC algorithms and explained a statistical point of view of the two algorithms, which provides the main motivation of this paper. In Section \ref{sec: two examples},  two examples are introduced as a comparison between the two algorithms. In the first example, the exchange algorithm performs better than the Pseudo-marginal Monte Carlo and in the second example, vise versa.  In Section \ref{sec: Pseudo-marginal Exchange},  we propose a new algorithm: Multi-armed Bandit MCMC (MABMC) which is a combination of the two RMCMC algorithms, obtaining higher acceptance probability.

\section{Review of the two RMCMC algorithms}\label{review}
\subsection{Pseudo-marginal Monte Carlo (PMC)}
To tackle the problem of the unknown ratio $\frac{Z(\theta)}{Z(\theta')}$, M{\o}ller et al. \cite{moller2006efficient} introduced an auxiliary variable $y$ which also takes value on the state space $\mathcal{X}$, the joint distribution is designed to be:
\[
\pi(x,y,\theta) = \pi(y|x,\theta) \frac{f_\theta(x)}{Z(\theta)}\pi(\theta)
\]
and therefore the previous joint distribution $\pi(x,\theta)$ is unaffected. Therefore, if we could sample from $\pi(y,\theta|x)$,  the marginal distribution of $\theta$ would be our target $\pi(\theta|x)$. The Pseudo-marginal Monte Carlo algorithm is designed as follows:

\begin{algorithm}
\caption{Pseudo-marginal Monte Carlo Algorithm}\label{alg:pseudo}
\hspace*{\algorithmicindent} \textbf{Input:} initial setting $(\theta, y)$, number of iterations $T$ \\
\begin{algorithmic}[1]
\For{$t= 1,\cdots T$}
\State Generate $\theta'\sim q(\theta'|\theta)$
\State Generate $y' \sim p_{\theta'}(y')= f_{\theta'}(y')/Z(\theta') $
\State Compute $$a =  \frac{\pi(\theta')q(\theta|\theta')f_{\theta'}(x)}{\pi(\theta)q(\theta'|\theta)f_{\theta}(x)}\cdot \frac{f_\theta(y)\pi(y'|x,\theta')}{f_{\theta'}(y')\pi(y|x,\theta)}$$
\State Draw $r \sim \text{Uniform}[0,1]$
\State \textbf{If} $(r< a)$ \textbf{then} set $(\theta, y)= (\theta',y')$.
\EndFor
\end{algorithmic}
\end{algorithm}

The auxiliary density $\pi(y|x,\theta)$ can be chosen by an arbitrary distribution. The only requirement for  $\pi(y|x,\theta)$ is that, for every $\theta$, the support of $\pi(y|x,\theta)$ contains the support of $p_{\theta}(x)$. For example, potential choice for $\pi(y|x,\theta)$ is, uniform distribution on $[0,1]$ if $\mathcal X = [0,1]$, or normal distribution with mean $\theta$ and variance $1$ if $\mathcal X = \bR$, or $\pi(y|x,\theta) = p_{\hat{\theta}}(y)$ where $\hat{\theta}$ is an estimator of $\theta$. The choice of $\pi(y|x,\theta)$ will not affect the correctness of the pseudo-marginal algorithm, but will definitely have a strong impact on the efficiency of the Markov chain. Detailed discussions and suggestions in choosing a proper
auxiliary density can be found in M{\o}ller et al. \cite{moller2006efficient}

PMC can be viewed as a M-H algorithm on the space $\Theta\times \mathcal{X}$, with transition kernel $q(\theta',y'|\theta, y) = q(\theta'|\theta) p_{\theta'}(y')$. With this choice of transition kernel, the acceptance ratio of the M-H algorithm becomes
\begin{align*}
a &= \frac{\pi(\theta',y'|x) } {\pi(\theta,y|x)}\cdot \frac{q(\theta|\theta') p_{\theta}(y)}{q(\theta'|\theta) p_{\theta'}(y')} \\
  & = \frac{\pi(\theta')f_{\theta'}(x)\pi(y'|x,\theta')}{\pi(\theta)f_\theta(x)\pi(y|x,\theta)}\cdot \frac{Z(\theta)}{Z(\theta')}\cdot \frac{q(\theta|\theta')}{q(\theta'|\theta)}\cdot \frac{f_\theta(y)}{f_{\theta'}(y')}\cdot \frac{Z(\theta')}{Z(\theta)}\\
  & = \frac{\pi(\theta')q(\theta|\theta')f_{\theta'}(x)}{\pi(\theta)q(\theta'|\theta)f_{\theta}(x)}\cdot \frac{f_\theta(y)\pi(y'|x,\theta')}{f_{\theta'}(y')\pi(y|x,\theta)} .\\
\end{align*}

Therefore the acceptance ratio does not depend on the unknown term $\frac{Z(\theta)}{Z(\theta')}$ .

One of the most important assumptions we made here is doing exact sampling from $p_{\theta'}(y')= f_{\theta'}(y')/Z(\theta')$ (Step 2 in Algorithm \ref{alg:pseudo}). As $Z(\theta')$ is unknown, this step is not easy and often not doable. Surprisingly, perfect sampling without knowing the normalizing constant is sometimes still possible using the `coupling from the past' method, see \cite{propp1996exact} for details.
However, in more cases, we usually establish another Markov chain with stationary distribution $p_{\theta'}(y')$ as an approximation in practice. 
\subsection{Modified Pseudo-marginal Monte Carlo (MPMC)}
The state space of PMC is $\Theta\times\mathcal{X}$, while the state space of exchange algorithm (described in Section \ref{subsec:exchange}) is $\Theta$. Therefore we provide a modified version of PMC, which is essentially the same as PMC but with state space $\Theta$, making it possible for comparing the two algorithms and incorporating them together.

The modified Pseudo-marginal Monte Carlo is designed as follows:
\newpage
\begin{algorithm}
\caption{Modified Pseudo-marginal Monte Carlo Algorithm}\label{alg:mpseudo}
\hspace*{\algorithmicindent} \textbf{Input:} initial setting $\theta$, number of iterations $T$ \\

\begin{algorithmic}[1]

\For{$t= 1,\cdots T$}
\State Propose $\theta'\sim q(\theta'|\theta)$
\State Propose $y \sim \pi(y|x,\theta)$
 and  $y' \sim p_{\theta'}(y')= f_{\theta'}(y')/Z(\theta') $
\State Compute $$a =  \frac{\pi(\theta')q(\theta|\theta')f_{\theta'}(x)}{\pi(\theta)q(\theta'|\theta)f_{\theta}(x)}\cdot \frac{f_\theta(y)\pi(y'|x,\theta')}{f_{\theta'}(y')\pi(y|x,\theta)}$$
\State Draw $r \sim \text{Uniform}[0,1]$
\State \textbf{If} $(r< a)$ \textbf{then} set $\theta = \theta'$.
\EndFor
\end{algorithmic}
\end{algorithm}

Before proving that MPMC is a valid Monte Carlo algorithm, we first discuss the difference between PMC and MPMC. For PMC, in each step there is an attempted move from $(y,\theta)$ to $(y', \theta')$ according to the kernel 
$q(\theta'|\theta) p_{\theta'}(y')$, and the acceptance ratio $a$ is calculated according to the M-H algorithm, therefore randomness only occurs in the process of generating a new proposal $(y',\theta')$. 

For MPMC, however, in each step the attempted move is proposed from $\theta$ to $\theta'$ according to $q(\theta'|\theta)$. Then two  auxiliary variables $(y, y') \sim \pi(y|x,\theta) p_{\theta'}(y')$ are generated and the corresponding acceptance ratio depends on the value of random variables $y, y'$. Randomness comes not only from the proposal step, but also from the procedure of calculating the acceptance ratio, which is different from PMC, or other standard M-H type algorithms. Those M-H algorithms with randomized acceptance ratio would be referred to as \textbf{randomized Markov chain Monte Carlo (RMCMC)}. MPMC and exchange algorithms are two typical examples of RMCMC algorithms.

 Now we prove the basic property of MPMC:

\begin{lemma} \label{lemma: detailed balance for MPMC}
The Modified Pseudo-marginal Monte Carlo Algorithm satisfies the detailed balance, i.e., 

\[
\pi(\theta |x) p(\theta \rightarrow \theta) = \pi(\theta'|x) p(\theta' \rightarrow \theta)  
\]
\end{lemma}
\begin{proof}
First we calculate the transition probability $p(\theta\rightarrow \theta')$, notice that to move from $\theta$ to $\theta'$, one has to first propose $\theta'$ according to the density $q(\theta‘|\theta)$ and then accept, and the acceptance probability depends on $y$ and $y'$ thus we need to take expectation with respect to $(y,y'))$. Therefore,
\begin{align*}
p(\theta' \rightarrow \theta')  & = q(\theta'|\theta) \bE_{y,y'}\min\Bigg\{\bigg( \frac{\pi(\theta')q(\theta|\theta')f_{\theta'}(x)}{\pi(\theta)q(\theta'|\theta)f_{\theta}(x)}\cdot \frac{f_\theta(y)\pi(y'|x,\theta')}{f_{\theta'}(y')\pi(y|x,\theta)}\bigg), 1 \Bigg\} \\
& =  q(\theta'|\theta)\int \int \min\Bigg\{\bigg( \frac{\pi(\theta')q(\theta|\theta')f_{\theta'}(x)}{\pi(\theta)q(\theta'|\theta)f_{\theta}(x)}\cdot \frac{f_\theta(y)\pi(y'|x,\theta')}{f_{\theta'}(y')\pi(y|x,\theta)}\bigg), 1 \Bigg\}  p_{\theta'} (y')\pi(y|x,\theta) d y d y'\\
& =   \int \int  \min\Bigg\{ \frac{\pi(\theta')q(\theta|\theta')p_{\theta'}(x)}{\pi(\theta) p_{\theta}(x)}\cdot p_\theta(y)\pi(y'|x,\theta'), p_{\theta'}(y') q(\theta'|\theta)\pi(y|x,\theta) \Bigg\} d y d y'.\\
\end{align*}
So we have
\begin{align*}
\pi(\theta |x) p(\theta \rightarrow \theta')  & = \frac{1}{\pi(x)} \pi(\theta) p_\theta(\theta|x) p(\theta \rightarrow \theta') \\
& = \frac 1 {\pi(x)} \int \int  \min\Bigg\{ \pi(\theta')p_{\theta'}(x)q(\theta|\theta') p_\theta(y)\pi(y'|x,\theta'),\\
&\qquad \qquad\pi(\theta) p_{\theta}(x) q(\theta'|\theta)p_{\theta'}(y') \pi(y|x,\theta) \Bigg\} dy dy' \\
& = \frac 1 {\pi(x)} \int \int  \min\Bigg\{ \pi(\theta')p_{\theta'}(x)q(\theta|\theta') p_\theta(y')\pi(y|x,\theta'),\\
&\qquad\qquad\pi(\theta) p_{\theta}(x) q(\theta'|\theta)p_{\theta'}(y) \pi(y'|x,\theta) \Bigg\} dy dy' \\
& = \pi(\theta'|x)p(\theta'\rightarrow \theta )
\end{align*}
The last equality comes from the fact that for any integrable function $f(y,y')$, 
\[
\int \int f(y, y') dy dy' = \int \int f(y',y) dy dy'
\]
\end{proof}

Lemma \ref{lemma: detailed balance for MPMC} implies MPMC constructs a reversible Markov chain with stationary distribution $\pi(\theta|x)$. 

It is not hard to show that (essentially one-step Janson's inequality), comparing with the original M-H chain (assuming the normalizing constant is known), the MPMC chain is less statically efficient. For all $\theta, \theta \in \Theta$, $a_{MPMC}(\theta, \theta') \leq a_{MH}(\theta, \theta')$. This puts MPMC chain below M-H chain in the ordering of Peskun \cite{peskun1973optimum}, although the M-H chain is not achievable here in our case.

The convergence property of MPMC requires more careful analysis. Nicholls et al. gives useful results on the convergence results for randomized MCMC \cite{nicholls2012coupled}. In our case, briefly speaking, if the original M-H chain  is uniformly ergodic or geometrically ergodic and the ratio \[\frac{f_\theta(y)\pi(y'|x,\theta')}{f_{\theta'}(y')\pi(y|x,\theta)}\]
is upper and lower bounded by a positive number, then so is the MPMC chain.

\subsection{Single Variable Exchange Algorithm (SVE)}\label{subsec:exchange}
The exchange algorithm is another RMCMC algorithm, which is similar to MPMC. However, the acceptance ratio is calculated (estimated) in a different way.

\begin{algorithm}
\caption{Exchange Algorithm}\label{alg:Exchange}
\hspace*{\algorithmicindent} \textbf{Input:} initial setting $\theta$, number of iterations $T$ \\

\begin{algorithmic}[1]

\For{$t= 1,\cdots T$}
\State Generate $\theta'\sim q(\theta'|\theta)$
\State Generate an auxiliary variable $w\sim f_{\theta'}(w)/Z(\theta')$
\State Compute $$a =  \frac{\pi(\theta')q(\theta|\theta')f_{\theta'}(x)}{\pi(\theta)q(\theta'|\theta)f_{\theta}(x)}\cdot \frac{f_\theta(w)}{f_{\theta'}(w)}$$
\State Draw $r \sim \text{Uniform}[0,1]$
\State \textbf{If} $(r< a)$ \textbf{then} set $\theta =\theta'$
\EndFor
\end{algorithmic}
\end{algorithm}
For SVE, in each step there is an attempted move from $\theta$ to $\theta'$ according to the same transition kernel $q(\theta'|\theta)$, however, SVE only generates one auxiliary variable $w \sim p_{\theta'}(w)$ and the acceptance ratio depends on $w$. SVE also preserves detailed balance.
\begin{lemma} \label{lemma: detailed balance for SVE}
The Single variable exchange algorithm satisfies the detailed balance, i.e., 

\[
\pi(\theta |x) p(\theta \rightarrow \theta) = \pi(\theta'|x) p(\theta' \rightarrow \theta)  
\]
\end{lemma}

The proof is very similar to \ref{lemma: detailed balance for SVE} and can be found in \cite{diaconis2018bayesian}.

Similar results on Peskun's ordering and convergence rate can be established for SVE, but it will be omitted here as this is not of our main focus in this paper.
\newpage
\subsection{ MPMC versus SVE, a statistical point of view}
In the abstract of Murray's exchange algorithm paper \cite{murray2012mcmc}, the authors claimed that SVE achieves better acceptance probability than PMC, and is justified by numerical simulation. This motivates us to raise the following question:

\begin{itemize}
    \item Is it possible to compare PMC and SVE theoretically?
\end{itemize}

As the state spaces of PMC and SVE are different, it makes more sense to compare SVE with MPMC. In this part, we provide a statistics point of view of SVE and PEMC, which also provides intuitions for future sections.

Recall that in standard M-H algorithm, given the stationary distribution $\pi(\theta|x) \propto \pi(\theta) \frac{f_\theta(x)}{Z(\theta)}$ and transition kernel $q(\theta|\theta')$, the acceptance ratio is 
\[
a = \frac{\pi(\theta')f_{\theta'}(x)q(\theta|\theta')}{\pi(\theta)f_{\theta}(x)q(\theta'|\theta)}\frac{Z(\theta)}{Z(\theta')}.
\]

All the terms can be computed except the ratio of unknown normalizing constants $\frac{Z(\theta)}{Z(\theta')}$. The obvious idea is to find an estimator of $Z(\theta)/Z(\theta')$ and plug it into the expression of acceptance ratio. This is widely used in practice, however, as mention in Section \ref{sec:introduction}, such estimators without other constraint will break the detailed balance and the corresponding Markov chain is not guaranteed to converge to the desired stationary distribution. Heuristically speaking, the idea of two RMCMC algorithms (MPMC and SVE) is to find a `good' estimator of $\frac{Z(\theta)}{Z(\theta')}$. The word `good' here means the estimator should preserve detailed balance of the resulting chain. It will soon be clear that the only difference between MPMC and SVE is that they use different estimators (denoted by $\ampmc$ and $\asve$, respectively) to estimate acceptance ratio $a$.

To be specific, in MPMC, the ratio 
$\frac{Z(\theta)}{Z(\theta')}$ is estimated by:
\begin{align*}
\frac{f_\theta(y) \pi(y'|x,\theta')}{f_{\theta'}(y')\pi(y|x,\theta)} \qquad \text{where} \quad  (y,y')|\theta,\theta' \sim \pi(y|x,\theta)\cdot  p_{\theta'}(y').
\end{align*}

Therefore the resulting randomized acceptance ratio is given by:
\[
\ampmc =  \frac{\pi(\theta')q(\theta|\theta')f_{\theta'}(x)}{\pi(\theta)q(\theta'|\theta)f_{\theta}(x)}\cdot \frac{f_\theta(y)\pi(y'|x,\theta')}{f_{\theta'}(y')\pi(y|x,\theta)}。
\]

$\ampmc$ is unbiased since
\begin{align*}
\bE_{(y,y')}\Bigg[\frac{f_\theta(y) \pi(y'|x,\theta')}{f_{\theta'}(y')\pi(y|x,\theta)} \Bigg] & = \int \frac{f_\theta(y)\pi(y'|x,\theta')}{f_{\theta'}(y')\pi(y|x,\theta)} \pi(y|x,\theta)\cdot  p_{\theta'}(y') dy dy'  \\
& = \frac{Z(\theta)}{Z(\theta')} \int p_\theta(y) \pi(y'|x,\theta') dy dy'\\
& = \frac{Z(\theta)}{Z(\theta')} ,
\end{align*}
as we proved in Lemma \ref{lemma: detailed balance for MPMC}, unbiasness preserves the detailed balance of MPMC, which guarantees the asymptotic exactness of MPMC algorithm.

Similarly, for SVE, the ratio $\frac{Z(\theta)}{Z(\theta')}$ is estimated by 
\begin{align*}
    \frac{f_\theta(w)}{f_{\theta'}(w)} \qquad\text{where}\quad w|\theta,\theta'\sim p_{\theta'}(w).
\end{align*}

Therefore the resulting randomized acceptance ratio is given by:
$$\asve =  \frac{\pi(\theta')q(\theta|\theta')f_{\theta'}(x)}{\pi(\theta)q(\theta'|\theta)f_{\theta}(x)}\cdot  \frac{f_\theta(w)}{f_{\theta'}(w)}$$

$\asve $ is clearly unbiased since
\begin{align*}
    \bE_w\Bigg[\frac{f_\theta(w)}{f_{\theta'}(w)}\Bigg] & = \int \frac{f_\theta(w)}{f_{\theta'}(w)} p_{\theta'}(w) dw \\
    & = \frac{Z(\theta)}{Z(\theta')} \int p_{\theta}(w) dw \\
    & = \frac{Z(\theta)}{Z(\theta')},
\end{align*}
 and again, unbiaseedness guarantees detailed balance.
\begin{remark}
This not necessary implies unbiasedness is the sufficient and necessary condition for designing an estimator of $\frac{Z(\theta)}{Z(\theta'}$.
\end{remark}

To summarize, given an attempted move from $\theta$ to $\theta'$, the acceptance ratio $a(\theta,\theta')$ is estimated by two unbiased estimators, $\ampmc$ and $\asve$. Then the acceptance probability $r(\theta, \theta)$ is estimated by 
\[
\rmpmc = \min\{\ampmc, 1\}
\]
and 
\[
\rsve = \min\{\asve, 1\}
\]
respectively.

Therefore comparing the two algorithms is equivalent to comparing the performance of the two estimators. As both $\ampmc$ and $\asve$ are unbiased, the performance only depends on the variance of the two estimators. The following theorem characterize the relative mean square error for both estimators by Pearson Chi-square distances $\chi_P$.

\begin{theorem} \label{thm: chi-square characterization}
Let 
\begin{align*}
    \text{RE}(\hat a) \doteq  \frac{(\hat a - a)^2}{a^2}
\end{align*}
be the relative mean square error of estimator $\hat a$, then we have 
\begin{align*}
\text{RE}(\asve) = \chi_P(p_{\theta'}, p_\theta) 
\end{align*}
and
\begin{align*}
\text{RE}(\ampmc) = \chi_P\big(p_{\theta'}(y')\pi(y|x,\theta), p_\theta(y)\pi(y'|x,\theta')\big),
\end{align*}
where $$\chi_p(f, g) = \frac{(f(x) - g(x))^2}{f(x)} dx.$$
\end{theorem}

\begin{proof}
For SVE, we have
\begin{align*}
    \text{RE}(\asve) & = \frac{(\asve - a)^2}{a^2} \\
    & = \int \frac{\Big(f_\theta(w)/f_{\theta'}(w) - Z(\theta)/ Z(\theta')\Big)^2 } {\Big(Z(\theta)/ Z(\theta')\Big)^2} p_{\theta'}(w) dw \\
    & = \int (\frac{p_\theta(w)}{p_{\theta'}(w)} - 1 )^2 p_{\theta'}(w) dw \\
    & = \int \frac{\big(p_\theta(w) - p_{\theta'}(w)\big)^2}{p_{\theta'}(w)} dw\\
    & =  \chi_P(p_{\theta'}, p_\theta).
\end{align*}

Similarly, for MPMC
\begin{align*}
    \text{RE}(\ampmc) & = \frac{(\ampmc - a)^2}{a^2} \\
     & = \int \frac{\Big(f_\theta(y) \pi(y'|x,\theta')/f_{\theta'}(y') \pi(y|x,\theta) - Z(\theta)/ Z(\theta')\Big)^2 } {\Big(Z(\theta)/ Z(\theta')\Big)^2} f_{\theta'}(y') \pi(y|x,\theta) dy dy'\\
     & = \int (\frac{p_\theta(y) \pi(y'|x,\theta')}{p_{\theta'}(y')\pi(y|x,\theta)} -1)^2 p_{\theta'}(y')\pi(y|x,\theta)dy dy'\\
     & = \chi_P\big(p_{\theta'}(y')\pi(y|x,\theta), p_\theta(y)\pi(y'|x,\theta')\big),
\end{align*}
as desired.
\end{proof}

Theorem \ref{thm: chi-square characterization} reveals a significant difference between MPMC and SVE. For MPMC, the choice of $\pi(\cdot|x,\theta)$ would influence the corresponding Peason Chi-square distance, and thus has a stong impace on the efficiency of the Markov chain. The optimal choice of $\pi(\cdot|x, \theta$ is clearly $p_\theta(\cdot)$ itself, such an estimator has $0$ variance and it makes Algorithm \ref{alg:M-H} and \ref{alg:mpseudo} agrees, but is impractical in our case. In practice, this suggests us to choose $\pi(\cdot|x,\theta)$ which is as close to $p_\theta$ as possible. 

For SVE, the variance is controlled by the Pearson Chi-square distance between $p_\theta$ and $p_{\theta'}$.
Roughly speaking, when $p_\theta$ is close to $p_{\theta'}$ (which is often equivalent to $\theta$ is close to $\theta'$), the SVE estimator tends to perform well. However, when $p_\theta$ is far away from $p_{\theta'}$, the SVE estimator may perform poorly due to the large variance.

Therefore, with a properly chosen $\pi(\cdot|x,\theta)$, it is reasonable to have the following heuristics:
\begin{itemize}
    \item When the proposed $\theta'$ is close to $\theta$, the SVE estimator would have better performance, resulting in a higher acceptance probability.
    \item When $\theta'$ is far away from $\theta$, the MPMC estimator would outperform SVE, and one should choose MPMC if possible.
\end{itemize}

The above heuristics suggests it is not possible to conclude that one algorithm dominates the other in all cases. In the next section  two concrete examples will be provided to justify our intuition. Meanwhile, in each step of M-H algorithm, a new $\theta'$ is proposed based on the previous $\theta$, this motivates us to find a method for choosing between MPMC and SVE adaptively, which is described in detail in Section \ref{sec: Pseudo-marginal Exchange}. 

\begin{remark}
In this paper, it is of our main focus to choose between MPMC and SVE. But the methodology proposed in Section \ref{sec: Pseudo-marginal Exchange} is general for all RMCMC algorithms with the same transition kernel. It would be very interesting to construct other RMCMC algorithms which preserves detailed balance. 
\end{remark}

\section{Two concrete examples}\label{sec: two examples}
In this section we will give two concrete examples, and argue that it is impossible to claim that one algorithm would always works better than the other.
\subsection{The first example}\label{eg:first}
Let $\mathcal  X$ be the space with two points, i.e.,  $\mathcal X = \{0,1\}$. Therefore the probability measure on $\mathcal X$ are Bernoulli distributions. Let the parameter space $\Theta$ also consists two points, $\Theta = \{a =0.7, b = 0.6\}$. Here $\Prob_{a}$ corresponds to the probability distribution on $\mathcal X$:
\[
\Prob_a(X = 1) = 0.7 \qquad \Prob_a(X= 0) = 0.3.
\]

Similarly $\Prob_b$ corresponds to the probability distribution  on $\mathcal X$:
\[
\Prob_b(X = 1) = 0.6 \qquad \Prob_b(X= 0) = 0.4.
\]

The prior distribution is chosen to be the uniform distribution over $\Theta$, and suppose the data $x$ equals $1$. Therefore, after simple calculation, the true posterior density would be:
\[
\Prob (\theta = a| x) = \frac 7 {13} \qquad \Prob (\theta = b|x) = \frac 6 {13}.
\]

For both algorithms, the transition probability $q(\cdot|\cdot)$ is the uniform distribution over $\Theta$. To make calculation easier, we choose uniform distribution over $\mathcal X$ as the conditional distribution $\pi(y|x,\theta)$, which is independent with data and parameter.

Now we are ready to calculate the transition probability of both algorithms, we will use $\Prob_{\text{SVE}}$ to denote the transition probability using exchange algorithm, and use  $\Prob_{\text{MPMC}}$ to denote the transition probability using modified Pseudo-marginal Monte Carlo. The transition probabilities can be calculated as follows:

\begin{align*}
 \Prob_{\text{SVE}}(\theta' =b| \theta =a) &= q(b|a)\cdot [ \Prob_b(w = 0) \cdot \min \{\frac{\Prob_b(x)}{\Prob_a(x)}\cdot \frac{\Prob_a(w)}{\Prob_b(w)} ,1\} + \Prob_b(w = 1) \cdot \min \{\frac{\Prob_b(x)}{\Prob_a(x)}\cdot \frac{\Prob_a(w)}{\Prob_b(w)} ,1\}]  \\
 & = \frac 12 \cdot [0.4 \cdot \min\{\frac 67 \times\frac 34, 1\} + 0.6 \cdot \min\{\frac 67 \times\frac 76, 1\}] \\
 & = \frac 37
\end{align*}

\begin{align*}
 \Prob_{\text{SVE}}(\theta' =a| \theta =b) &= q(a|b)\cdot [ \Prob_a(w = 0) \cdot \min \{\frac{\Prob_a(x)}{\Prob_b(x)}\cdot \frac{\Prob_b(w)}{\Prob_a(w)} ,1\} + \Prob_a(w = 1) \cdot \min \{\frac{\Prob_a(x)}{\Prob_b(x)}\cdot \frac{\Prob_b(w)}{\Prob_a(w)} ,1\}]  \\
 & = \frac 12 \cdot [0.3 \cdot \min\{\frac 76 \times\frac 43, 1\} + 0.7 \cdot \min\{\frac 76 \times\frac 67, 1\}] \\
 & = \frac 12
\end{align*}
Similarly, 
\begin{align*}
\Prob_{\text{MPMC}}    (\theta' =b| \theta =a) & =q(b|a)\cdot \bE_{y,y'}\big[\min\{ \frac{\Prob_b(x)}{\Prob_{a}(x)}\cdot \frac{\Prob_a(y)}{\Prob_{b}(y')} , 1 \}\big] \\
 & = \frac 12 \times  (\frac 3 {10} \times  \frac 67 \times  \frac 76 + \frac 3{10} \times \frac 67 \times  \frac 36 + \frac 15\times  \min\{\frac 67 \times  \frac 74, 1\}  + \frac 1 {5 }\times \frac 67 \times  \frac 34 )\\
 & = \frac {11}{28}
\end{align*}

\begin{align*}
\Prob_{\text{MPMC}}    (\theta' =a| \theta =b) & =q(a|b)\cdot \bE_{y,y'}\big[\min\{ \frac{\Prob_a(x)}{\Prob_{b}(x)}\cdot \frac{\Prob_b(y)}{\Prob_{a}(y')} , 1 \}\big] \\
 & = \frac 12 \times  (\frac 7 {20} \times  \frac 76 \times  \frac 67 + \frac 7{20} \times \frac 76 \times  \frac 47 + \frac 3 {20 }\times   \min\{\frac 76 \times  \frac 63, 1\}  + \frac 3 {20 }\times \min\{\frac 76 \times  \frac 43, 1\})\\
 & = \frac {55}{120}.
\end{align*}

Therefore we have $\Prob_{pm}(\theta'|\theta) < \Prob_{ex}(\theta'|\theta)$ for any $\theta' \neq \theta$ and thus the exchange algorithm has higher acceptance transition probability, which means exchange algorithm will converge faster.

\subsection{The second example}\label{eg:second}
The second example is designed to show that Pesudo-marginal Monte Carlo may perform better than exchange algorithm. 
In this example we take $\mathcal  X$ be the space with three points, i.e.,  $\mathcal X = \{0,1, 2\}$, the parameter space $\Theta$ also consists two points, $\Theta = \{a, b \}$. Here $\Prob_{a}$ corresponds to the probability distribution on $\mathcal X$:
\[
\Prob_a(X = 0) = 0.1 \qquad \Prob_a(X= 1) = 0.8 \qquad \Prob_a(X = 2) = 0.1.
\]

Similarly $\Prob_b$ corresponds to the probability distribution  on $\mathcal X$:
\[
\Prob_b(X = 0) = 0.8 \qquad \Prob_b(X= 1) = 0.1 \qquad \Prob_b(X = 2) = 0.1.
\]

The prior distribution are chosen to be the uniform distribution over $\Theta$, and suppose the data $x$ equals $2$. Therefore, after simple calculation, the true posterior density would be:
\[
\Prob (\theta = a| x) = \frac 12 \qquad \Prob (\theta = b|x) = \frac 12.
\]

Similar to the previous example, the transition probability $q(\cdot|\cdot)$ is the uniform distribution over $\Theta$.  The conditional distribution $\pi(y|x,\theta)$ is designed to be the uniform distribution over $\mathcal X$, which is independent with data and parameter.

The last thing we need to specify is the way we initialize variable $y$ in the pseudo-marginal Monte Carlo algorithm, we simply draw $y$ with $\Prob(y = 0) =\Prob (y = 1) = \Prob(y = 2) = \frac 13$. 

In this setting, we are ready to calculate the transition probability of both algorithms, we will still use $P_{\text{SVE}}$ to denote the transition probability using exchange algorithm, and use  $P_{\text{MPMC}}$ to denote the transition probability using pseudo marginal Monte Carlo. The transition probabilities can be calculated as follows:
\begin{align*}
 \Prob_{\text{SVE}}(\theta' =b| \theta =a) &= q(b|a)\cdot \bE _w\Bigg[\min \{\frac{\Prob_b(x)}{\Prob_a(x)}\cdot \frac{\Prob_a(w)}{\Prob_b(w)} ,1\}\Bigg]\\
 & = \frac 12 \cdot [0.8 \cdot \min\{\frac 18 , 1\} + 0.1 \cdot \min\{8, 1\}  + 0.1] \\
 & = \frac{3}{20}
\end{align*}
By symmetry, we also have 
\[
\Prob_{\text{SVE}}(\theta' =a| \theta =b) = \frac {3}{20}. 
\]
Similarly,

\begin{align*}
\Prob_{\text{MPMC}}    (\theta' =a| \theta =b) & =q(a|b)\cdot \bE_{y,y'}\big[\min\{ \frac{\Prob_a(x)}{\Prob_{b}(x)}\cdot \frac{\Prob_b(y)}{\Prob_{a}(y')} , 1 \}\big] \\
 & = \frac 12 \times  (0.1 + 0.1 + 0.8\times \frac 13   + 0.8 \times \frac 13 \times \frac 18 + 0.8 \times \frac 13 \times \frac 18 )\\
 & = \frac {4}{15}.
\end{align*}

and 
\[
\Prob_{\text{MPMC}}(\theta' =a| \theta =b) = \frac {4}{15}.
\]

Therefore we have $\Prob_{pm}(\theta'|\theta) > \Prob_{ex}(\theta'|\theta)$ for any $\theta' \neq \theta$ and thus the Pseudo-marginal algorithm has higher acceptance transition probability, which means it will converge faster.

\subsection{Intuition and discussion}

In this part we briefly talk about the intuition behind the two examples above. Comparing with the ideal Metropolis-Hastings algorithm, the difficult part comes from the unknown ratio of normalizing constant
\[
\frac{Z(\theta)}{Z(\theta')}.
\]

The idea of both algorithms is to generate an auxiliary variable, and use the new random variable to construct an estimator to estimate the ratio of normalizing constants. Therefore, the accuracy of the estimator determines the performance of the algorithm. The main difference between Pseudo-marginal Monte Carlo and exchange algorithm is, the exchange algorithm uses the estimator
\[
\frac{f_\theta(w)}{f_{\theta'}(w)}
\]

to estimate the ratio directly. However, the Pseudo-marginal Monte Carlo uses $f_\theta(y)\pi(y'|x,\theta') $
to estimate $Z(\theta)$ and uses 
$f_{\theta'}(y')\pi(y|x,\theta)$ to estimate $Z(\theta')$, then it uses the quotient as an estimator of $\frac{Z(\theta)}{Z(\theta')}$. Therefore, when the probability measure $\pi_{\theta}$ and $\pi_{\theta'}$ differs a lot (the case in the second example), then the exchange algorithm may perform poorly. But when the two measures are close, then the exchange algorithm may perform better than Pseudo-marginal Monte Carlo.

In the two examples above, one algorithm dominates the other for all pairs $(\theta, \theta')$. In real situations, however, the parameter space $\Theta$ may be much larger than our designed examples, thus the usually there exists two rigions, $R_1, R_2 \subset \Theta\times \Theta$. In $R_1$, PMC performs better than SVE, while in $R_2$, vice versa. Therefore, given a proposal move from $\theta$ to $\theta'$, it is natural to ask the following question:
\begin{itemize}
    \item Is it possible to find a `reasonable' way of choosing between PMC and SVE to improve the acceptance probability ?
\end{itemize}

This question will be answered in the next section.

\section{MABMC: A Multi-armed Bandit MCMC Algorithm}\label{sec: Pseudo-marginal Exchange}
\subsection{How to choose between algorithms `legally'?} \label{subsec:decision rule}
Now we are ready to incorporate the two algorithms together. To make things more general and concrete, this problem can be formulated as follows:

Given an attempt move from $\theta$ to $\theta'$,  and two valid estimates of $a(\theta, \theta)$, denoted by $\hat{a}_1(\theta,\theta')$, $\hat{a}_2(\theta,\theta')$. Can one find a decision rule $D(\theta,\theta') \in \{1,2\}$ such that the new Markov chain with transition probability $ q(\theta,\theta')\min\{a_D(\theta,\theta'), 1\}$ still preserves detailed balance and has higher acceptance probability than either algorithm?

It is worth mentioning that the seemingly obvious choice $\text{argmax}\{\hat{a}_1(\theta,\theta'), \hat{a}_2(\theta,\theta')\}$ is not valid, as this would break the detailed balance and thus the algorithm may not converge to the desired stationary distribution. It turns out that to preserve the detailed balance, the decision rule has to be `symmetric', i.e., $D(\theta, \theta') = D(\theta', \theta)$.

This problem is very similar to the multi-armed bandit problem in probability. In each iteration of the MCMC algorithm, an agent is facing the choose between $\hat a_1$ and $\hat a_2$, with the goal of increasing the acceptance probability. It is of our interest to design a reasonable policy (or make a decision rule) to get better performance than random guess.
\begin{definition}[Valid ratio]\label{def:valid ratio}
Let $\pi(\theta)$ be a probability density on the parameter space $\Theta$, which may be of the form $\pi(\theta) = \frac{f(\theta)}{Z(\theta)}$ where $Z(\theta)$ is intractable.  Let $q(\theta,\theta')$ be a transition kernel in the M-H algorithm. A (randomized) acceptance ratio $\hat{a} (\theta, \theta')$ is called \textbf{valid} if
it preserves the detailed balance with respect to stationary distribution $\pi(\theta)$, i.e.,
\begin{align*}
    \pi(\theta) q(\theta,\theta') \bE\min\{\hat{a}(\theta, \theta'), 1\}  =  \pi(\theta') q(\theta',\theta) \bE\min\{\hat{a}(\theta', \theta), 1\}.
\end{align*}
\end{definition}

\begin{example}
The acceptance ratio introduced in exchange algorithm (Alg \ref{alg:Exchange}) and modified Pseudo-marginal Monte Carlo (Alg \ref{alg:mpseudo}) are both valid.
\end{example}

\begin{definition}[Valid decision rule]
Given the target stationary distribution $\pi(\theta)$, the transition kernel $q(\theta, \theta')$ and two valid acceptance ratio $\hat{a}_1(\theta, \theta')$, $\hat{a}_2(\theta, \theta')$ . A decision rule $D: \Theta \times \Theta \rightarrow \{1, 2\}$ is called \textbf{valid} if the corresponding new acceptance ratio $\hat{a}_D (\theta, \theta')$ is valid.
\end{definition}

Intuitively, in each iteration of the M-H algorithm,  given an attempted move from $\theta$ to $\theta'$, the decision rule $D(\theta, \theta')$ helps one to choose between the two acceptance ratio adaptively, aiming for a higher acceptance probability while still preserving the detailed balance. The decision rule is implicitly random, since the acceptance ratio $\hat{a}_1(\theta, \theta')$, $\hat{a}_2(\theta, \theta')$ are random. 

The following example gives a simple, non-randomized decision rule.
\begin{example}[A simple valid decision rule]\label{eg: simple decision rule} The decision rule $D(\theta, \theta') \equiv 1$ is a valid decision rule. It corresponds to always choosing the first acceptance ratio for each attempted move. Similarly, $D(\theta, \theta') \equiv 2$ is also valid.
\end{example}
We could also define the `Bayes decision rule' and `inadmissible decision rule', which is similar to the definition of statistics in statistical decision theory.
\begin{definition}[Bayes decision rule]\label{def: optimal decision rule}
A valid decision rule $D$ is called \textbf{Bayes} if for any other valid decision rule $\tilde D$, 
\[
\bE{ \min \{\hat a_D, 1\}} \geq \bE{ \min \{\hat a_{\tilde D}, 1\}},
\]
where the expectation is taken over all the randomness (including the transition kernel $q$, the decision procedure, the estimation of the acceptance ratio $a$, and integrating over $\Theta \times \Theta$.
\end{definition}

This decision rule is called `Bayes' since the stationary distribution $\pi(\theta|x)$ together with $q(\theta'|\theta)$ can be regarded as a prior distribution on $\Theta \times \Theta$, and the Bayes decision decision rule $D$ maximized the average acceptance probability according to this prior. Sometimes one may be interested in a point-wise relationship between two decision rules, which motivates the following definition.

\begin{definition}[Inadmissible decision rule]\label{def:inadmissible decision rule}

A decision rule $D_1$ is called `inadmissible' if there exists another decision rule $D_2$, such that for all $(\theta, \theta')\in \Theta \times \Theta$, 
\[
\bE \min\{\hat a_{D_1}(\theta,\theta') , 1 \} \leq \bE \min\{\hat a_{D_2}(\theta,\theta') , 1 \},
\]
and the inequality is strict in at least one point $(\theta_0, \theta'_0)$. A decision rule which is not inadmissible is called \text{admissible} decision rule.
\end{definition}

A decision rule is inadmissible means one could find another decision rule which dominates the previous one. It is clear that a Bayes decision rule would be admissible, as by definition it maximizes the average acceptance probability. Though it is not necessary true that the Bayes decision rule will dominate other decision rules at every point. In general, the calculation of the Bayes decision rule is beyond our knowledge, therefore in this paper we will focus on finding a reasonable rule which is better than random guess, instead of finding the Bayes decision rule.

\begin{example}
In Section \ref{eg:first}, the decision rule $D =$ PMC is inadmissible, in section \ref{eg:second}, $D=$ SVE is inadmissible.
\end{example}

It is not hard to generalize the definition of valid decision rule to $n$ valid acceptance ratio, the formal definition is omitted.

\begin{example}[A max-min decision rule] \label{eg:max-min decision rule} The decision rule 

\[
D(\theta,\theta') = \argmax_{i \in 1,2} \{\min\{r_i(\theta,\theta') , r_i(\theta',\theta)\} \}
\]
where $r_i(\theta,\theta') = \bE \min\{\hat{a_i}(\theta,\theta'),1\} $, is valid.
\end{example}
The max-min decision rule is valid can be proved as a direct corollary after proving Theorem \ref{thm:valid decision rule}. This will be used as the decision rule for the Multi-armed Bandit MCMC (MABMC). The intuition behind the design will be explained later.

\begin{example}[An invalid decision rule] \label{eg:invalid decision} The decision rule
\begin{align*}
    D(\theta,\theta') = \argmax_{i \in 1,2} \{\hat{a}_i(\theta,\theta')\}
\end{align*}
is not valid, (as we will see later) the corresponding acceptance ratio $\hat{a}_D (\theta, \theta')$ is not valid.
\end{example}
\begin{theorem}\label{thm:valid decision rule}
The decision rule $D$ is valid if and only if it is `symmetric', i.e.,
\[
D(\theta,\theta') = D(\theta', \theta) \quad \text{for all} \quad (\theta, \theta') \in K,
\]
where $K$ is a symmetric subset of $\Theta\times \Theta$, defined by 
\[
K \doteq \{(\theta,\theta')\in \Theta\times \Theta: r_1(\theta,\theta')\neq r_2(\theta,\theta') \quad \text{or}\quad   r_1(\theta',\theta) \neq r_2(\theta',\theta)\}
\]
\end{theorem}
\begin{proof}
For each fixed $(\theta,\theta')$, the detailed balance equation requires
\[  \pi(\theta) q(\theta,\theta') \bE\min\{\hat{a}_D(\theta, \theta'), 1\}  =  \pi(\theta') q(\theta',\theta) \bE\min\{\hat{a}_D(\theta', \theta), 1\}.\]

i.e., 
\[
 r_D(\theta,\theta')/ r_D(\theta',\theta) = \pi(\theta') q(\theta',\theta) /\pi(\theta)q(\theta,\theta').
\]

Meanwhile, as $\hat a_1$, $\hat a_2$ are valid ratios, we have
\[
 r_1(\theta,\theta')/ r_1(\theta',\theta) =  r_2(\theta,\theta')/ r_2(\theta',\theta) =
 \pi(\theta') q(\theta',\theta) /\pi(\theta)q(\theta,\theta').
\]
Therefore, for $(\theta,\theta')\in K$, it is clear that 
$D(\theta,\theta') = D(\theta',\theta)$. As $D(\theta,\theta') = 1, D(\theta',\theta) = 2$ or 
$D(\theta,\theta') = 2, D(\theta',\theta) = 1$ would breaks the detailed balance.

For $(\theta,\theta')\in K^c$, as $r_1(\theta,\theta') = r_2(\theta,\theta')$ and $r_1(\theta',\theta) = r_2(\theta',\theta)$. The two ratios are indistinguishable, thus $D(\theta,\theta‘)$ could be either $1$ or $2$.
\end{proof}
\begin{remark}
It could be reasonable to assume that $K = \Theta\times \Theta$ in real situations. Since $\hat a_1 and \hat a_2$ are two different estimators, usually the condition $ r_1(\theta,\theta')\neq r_2(\theta,\theta')$  is satisfied naturally.
\end{remark}
Therefore, given a proposal from $\theta$ to $\theta'$, to preserve the detailed balance, the decision $D$ has to be then same as the `reversed proposal' from $\theta'$ to $\theta$. This implies the decision rule in Example \ref{eg:invalid decision} is invalid. 
\begin{corollary}
The max-min decision rule in Example \ref{eg:max-min decision rule} is valid as $D(\theta, \theta')$ is symmetric by design.
\end{corollary}

\begin{corollary}
The max decision rule in Example \ref{eg:invalid decision} is invalid as $D(\theta, \theta')$ is not symmetric.
\end{corollary}
\subsection{MABMC: Multi-armed Bandit MCMC, algorithm and intuition}\label{subsec: MABMC, algorithm and intuition}
The Multi-armed Bandit MCMC Algorithm(MABMC) incorporates the Pesudo-marginal Monte Carlo (PMC) and Exchange algorithm (SVE) adaptively, according to the max-min decision rule:
\begin{algorithm}
\caption{Multi-armed Bandit MCMC (MABMC)}\label{alg:MABMC}
\hspace*{\algorithmicindent} \textbf{Input:} initial setting $\theta$, number of iterations $T$ \\

\begin{algorithmic}[1]

\For{$t= 1,\cdots T$}
\State Propose $\theta'\sim q(\theta'|\theta)$
\State Generate auxiliary variables:  $y \sim \pi(y|x,\theta), y' \sim f_{\theta'}(y')/Z(\theta') , w\sim f_{\theta'}(w)/Z(\theta')$
\State Compute 
\begin{align*}
& a_1 =  \frac{\pi(\theta')q(\theta|\theta')f_{\theta'}(x)}{\pi(\theta)q(\theta'|\theta)f_{\theta}(x)}\cdot \frac{f_\theta(y)\pi(y'|x,\theta')}{f_{\theta'}(y')\pi(y|x,\theta)},~~ r_1 = \min\{a_1,1\} \\
& a_2 =  \frac{\pi(\theta')q(\theta|\theta')f_{\theta'}(x)}{\pi(\theta)q(\theta'|\theta)f_{\theta}(x)}\cdot \frac{f_\theta(w)}{f_{\theta'}(w)},~~ r_2 = \min\{a_2, 1\}
\end{align*}

\State Generate auxiliary variables:  $\tilde y \sim \pi(y|x,\theta'), \tilde y' \sim f_{\theta}(y')/Z(\theta) , \tilde w\sim f_{\theta}(w)/Z(\theta)$

\State Compute 
\begin{align*}
& \tilde a_1 =  \frac{\pi(\theta)q(\theta'|\theta)f_{\theta}(x)}{\pi(\theta')q(\theta|\theta')f_{\theta'}(x)}\cdot \frac{f_{\theta'}(\tilde y)\pi(\tilde y'|x,\theta)}{f_{\theta}(\tilde y')\pi(\tilde y|x,\theta')},~~ \tilde r_1 = \min\{\tilde a_1,1\} \\ \\
& \tilde a_2 =  \frac{\pi(\theta)q({\theta'}|\theta)f_{\theta}(x)}{\pi({\theta'})q(\theta|{\theta'})f_{\theta'}(x)}\cdot \frac{f_{\theta'}(\tilde w)}{f_{\theta}(\tilde w)}, ~~ \tilde r_2 = \min\{\tilde a_2, 1\}
\end{align*}
\State Choose $$D = \argmax_{i\in \{1,2\}} \min\{r_i, \tilde r_i\}$$
\State \textbf{If} $(D = 1)$ \textbf{then} \text{repeat Step $2-7$ in  Algorithm \ref{alg:mpseudo}}
\State \textbf{If} $(D = 2)$ \textbf{then} \text{repeat Step $2-7$ in  Algorithm \ref{alg:Exchange}}

\EndFor
\end{algorithmic}
\end{algorithm}

Roughly speaking, MABMC is a way of choosing between to algorithms adaptively, according to the max-min decision rule introduced in Section \ref{subsec:decision rule}. Theoretically any valid decision rule (also defined in Section \ref{subsec:decision rule}) would give a new Markov chain satisfying detailed balance. However, implementing simple decisions rules like Example \ref{eg: simple decision rule} would degenerate the algorithm to Algorithm \ref{alg:Exchange} or Algorithm \ref{alg:mpseudo} without essential improvement. It will be shown later that Algorithm \ref{alg:MABMC} do improve the average acceptance rate, comparing with Algorithm \ref{alg:Exchange} and Algorithm \ref{alg:mpseudo}.

To explain the intuition of MABMC, let's go back to the simple invalid example \ref{eg:invalid decision}. Given $\hat a_1(\theta, \theta')$ and $\hat a_2(\theta,\theta')$, the intuitively natural choice 

$$  D(\theta,\theta') = \argmax_{i \in 1,2} \{\hat{a}_i(\theta,\theta')\}$$

as described in example \ref{eg:invalid decision} is itself invalid. But based on this decision rule, we could do a `symmetrization' as described in example \ref{eg:max-min decision rule}, making it to be a valid decision rule without losing too much. Heuristically speaking, it should be true that 
\[
\argmax_{i \in 1,2} \{\min\{\hat{r_i}(\theta,\theta') , \hat{r_i}(\theta',\theta)\} \approx \argmax_{i \in 1,2} \{\hat{a}_i(\theta,\theta')\}.
\]

This heuristics comes from the fact that $a(\theta, \theta') = \frac{1}{a(\theta', \theta)}$ by design of M-H algorithm. Therefore, if $\hat a_1$, $\hat a_2$ are both reasonable estimates of $a$, then the maximum of  $\hat r_i(\theta,\theta')$ and $\hat r_i(\theta',\theta)$ should be approximately $1$, as $a(\theta,\theta')$ and $a(\theta', \theta)$ can not be both less than $1$. We have, 
\begin{align*}
\argmax_{i \in 1,2} \{\min\{\hat{r_i}(\theta,\theta') , \hat{r_i}(\theta',\theta)\} & \approx \argmax_{i \in 1,2} \{\hat{r}_i(\theta,\theta')\} \\
& \approx \argmax_{i \in 1,2} \{\hat{a}_i(\theta,\theta')\}.
\end{align*}
Therefore the max-min decision rule in example \ref{eg:max-min decision rule} should be close to the max decision rule in example \ref{eg:invalid decision}, but is itself valid.

\section{Numerical Example}
\subsection{Normal example}
We consider a concrete example for which all the computations can be done easily. Consider the problem of sampling from the posterior of $\theta$, which has likelihood $p_\theta(y) \sim \cN(\theta, \sigma^2)$ with a conjugate prior $\pi(\theta) \sim \cN(0, 1)$, standard calculation gives the posterior distribution
\[
\pi(\theta|x) \sim \cN(\frac{y}{1+\sigma^2}, \frac{\sigma^2}{1 + \sigma^2}).
\]
The likelihood has the form 
\[
p_\theta(y) = \frac{1}{\sqrt{2\pi\sigma^2}}e^{- \frac{(y - \theta)^2}{2\sigma^2}}
\]
which is tractable. However, we pretend as if the normalizing constant $\frac{1}{\sqrt{2\pi\sigma^2}} $ is unknown to us. The result of the average acceptance probabilities by MPMC, SVE and MABMC is reported.

For each $\sigma^2$ ranging from $0.1$ to $1$, MPMC, SVE and MABMC is implemented for $20000$ iterations respectively. The transition kernel $q(\theta'|\theta) \sim \cN (\theta, 1)$, meanwhile $\pi(y|x,\theta)$ is chosen to be a normal distribution $\cN(\theta + \frac 13, \sigma^2)$, $y = 1$, figure \ref{fig:MABMC_gaussian} reports the average acceptance probability for MABMC, MPMC and SVE respectively. 

Figure \ref{fig:MABMC_gaussian} shows MABMC achieves higher average acceptance probability than both MPMC and SVE for all choices of $\sigma^2$, which shows such an max-min decision rule do improve the performance of the corresponding Markov chain. In this artificial example, the performance of MPMC and SVE is similar (the blue curve and green curve are similar in the plot), if one algorithm always performs better than the other, then the decision rule will basically choose one algorithm, and thus the performance of MABMC will be similar to the better one. Meanwhile, while for all $\sigma^2$, MABMC performs better than MPMC and SVE, it is clear that the improvement decreases as $\sigma^2$ increases. This is natural since when $\sigma^2$ increases, the likelihood would be flatter and thus an attempting move is more likely to be accepted, thus all three algorithms would have higher acceptance probability and the improvement of MABMC would be less significant. It turns out that for more peaked distribution (corresponding to small values of $\sigma^2$, the advantage of MABMC is more significant.
\begin{figure}[htbp]
\includegraphics[width= \textwidth]{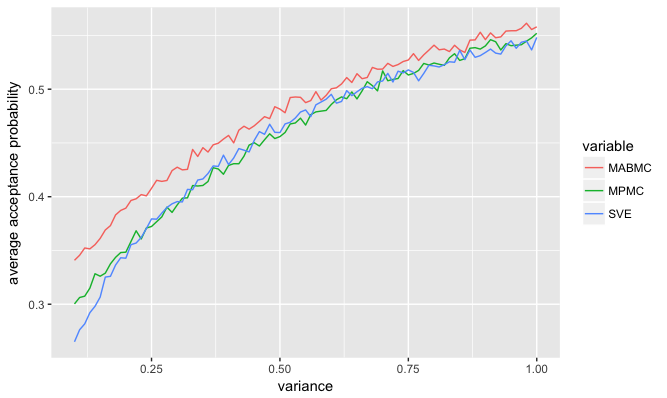}
\caption{Average acceptance probabilities as a function of $\sigma^2$ for a $\cN(\theta, \sigma^2)$ likelihood.}
\label{fig:MABMC_gaussian}
\end{figure}
\subsection{Ising Example}
We have also considered the 2-D Ising distribution on on a square lattice $\Lambda$ with $N$ sites. For each site $k \in\Lambda$ there is  a discrete variable $\sigma_k \in \{-1, +1\}$ representing the site's spin. A spin configuration $\sigma$ is an assignment of spin  value to each lattice site. The configuration probability is given by the Boltzmann distribution with parameter $\beta$:
\[
p_\beta(\sigma) = \frac{e^{-\beta H(\sigma)}}{Z(\beta)},
\]
where $H(\sigma) = -J\sum_{\langle i, j \rangle} \sigma_i \sigma_j $, the notation $\langle i, j \rangle $ indicates that sites $i$ and $j$ are nearest neighbors.
The normalizing constant
\[
Z(\beta) = \sum_\sigma e^{-\beta H(\sigma)}
\]
is intractable for large $N$, as the summation contains $2^{N^2}$ terms. Throughout the experiment, we set the interaction parameter $J = 0.1$, the number of sites $N = 10$.

For each $\beta$, we have generated $10$ Ising configurations using the Wolff algorithm \cite{wolff1989collective}, a normal prior was put on $\beta$ and we sample from posterior using  MPMC, SVE and MABMC respectively. The auxiliary density for MPMC is chosen to be $\pi(y|\text{data},\beta) = p_{\hat \beta}(y)$, where $\hat\beta$ is the  maximum pseudo-likelihood estimator (MPLE)  derived in \cite{besag1986statistical} for $\beta$. 

Figure \ref{fig:MABMC_Ising} shows, the MABMC algorithm overall outperforms better than SVE and MPMC algorithm. In particular, if SVE is significantly better than MPMC (or vice versa), MABMC would have acceptance probability close to the better one. If SVE and MPMC perform similarly, MABMC would be able to achieve a higher acceptance probability.
\begin{figure}[htbp]
	\includegraphics[width= \textwidth]{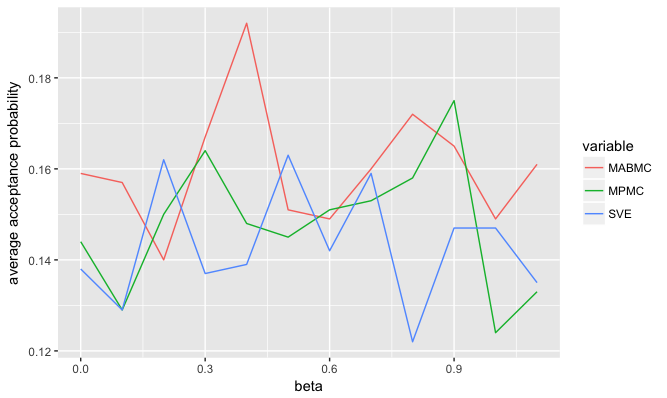}
	\caption{Average acceptance probabilities as a function of $\beta$ for  Ising model.}
	\label{fig:MABMC_Ising}
\end{figure}
\newpage
\section{Discussion}\label{sec:Discussion}
In this paper we have discussed the two most popular methods of tackling the doubly intractable problem -- Pseudo-marginal Monte Carlo (PMC) and exchange algorithm (SVE) in a statistical point of view. We have proposed MABMC which chooses between PMC and SVE for the `better' one in each step of the Markov chain iteration. It turns out that MABMC is easy to implement and achieves better average acceptance probability than PMC and SVE. However, many unknown questions remains:
\begin{itemize}
    \item MPMC and SVE are two special cases of RMCMC algorithms, aiming for solving the doubly intractable problems. However, MABMC algorithms works for general MCMC algorithms with randomized acceptance ratios, given the transition kernels and target distributions are the same. It would be interesting to find out the applications of MABMC in other disciplines other than the doubly intractable problem.
    \item The MABMC (min-max decision rule) comes from the intuition of modifying the natural invalid max decision rule (Example \ref{eg:invalid decision}), see \ref{subsec: MABMC, algorithm and intuition} for details. However, there is no guarantee that MABMC a is Bayesian decision rule, or even admissible (see Definition \ref{def: optimal decision rule}, \ref{def:inadmissible decision rule}), therefore there is still a large room to improve MABMC with respect to some metrics, or find another decision rule which dominates MABMC. 
    \item The convergence property of MABMC is unknown to us. It is clear that MABMC is less statistically efficient than the intractable original M-H algorithm. However, assuming the original M-H algorithm is uniformly ergodic or geometrically ergodic, it is worth investigating the condition that could ensure MABMC inherits such properties.
\end{itemize}

\newpage
\bibliographystyle{alpha}
\bibliography{bibliography}

\newcommand{\etalchar}[1]{$^{#1}$}
\begin{thebibliography}{MPRB06}

\bibitem[AFEB16]{alquier2016noisy}
Pierre Alquier, Nial Friel, Richard Everitt, and Aidan Boland.
\newblock Noisy monte carlo: Convergence of markov chains with approximate
  transition kernels.
\newblock {\em Statistics and Computing}, 26(1-2):29--47, 2016.

\bibitem[AR{\etalchar{+}}09]{andrieu2009pseudo}
Christophe Andrieu, Gareth~O Roberts, et~al.
\newblock The pseudo-marginal approach for efficient monte carlo computations.
\newblock {\em The Annals of Statistics}, 37(2):697--725, 2009.

\bibitem[AV{\etalchar{+}}15]{andrieu2015convergence}
Christophe Andrieu, Matti Vihola, et~al.
\newblock Convergence properties of pseudo-marginal markov chain monte carlo
  algorithms.
\newblock {\em The Annals of Applied Probability}, 25(2):1030--1077, 2015.

\bibitem[Bea03]{beaumont2003estimation}
Mark~A Beaumont.
\newblock Estimation of population growth or decline in genetically monitored
  populations.
\newblock {\em Genetics}, 164(3):1139--1160, 2003.

\bibitem[Bes74]{besag1974spatial}
Julian Besag.
\newblock Spatial interaction and the statistical analysis of lattice systems.
\newblock {\em Journal of the Royal Statistical Society: Series B
  (Methodological)}, 36(2):192--225, 1974.

\bibitem[Bes86]{besag1986statistical}
Julian Besag.
\newblock On the statistical analysis of dirty pictures.
\newblock {\em Journal of the Royal Statistical Society. Series B
  (Methodological)}, pages 259--302, 1986.

\bibitem[CS{\etalchar{+}}97]{chen1997monte}
Ming-Hui Chen, Qi-Man Shao, et~al.
\newblock On monte carlo methods for estimating ratios of normalizing
  constants.
\newblock {\em The Annals of Statistics}, 25(4):1563--1594, 1997.

\bibitem[DW18]{diaconis2018bayesian}
Persi Diaconis and Guanyang Wang.
\newblock Bayesian goodness of fit tests: a conversation for david mumford.
\newblock {\em arXiv preprint arXiv:1803.11251}, 2018.

\bibitem[GM98]{gelman1998simulating}
Andrew Gelman and Xiao-Li Meng.
\newblock Simulating normalizing constants: From importance sampling to bridge
  sampling to path sampling.
\newblock {\em Statistical science}, pages 163--185, 1998.

\bibitem[Lia10]{liang2010double}
Faming Liang.
\newblock A double metropolis--hastings sampler for spatial models with
  intractable normalizing constants.
\newblock {\em Journal of Statistical Computation and Simulation},
  80(9):1007--1022, 2010.

\bibitem[LJSL16]{liang2016adaptive}
Faming Liang, Ick~Hoon Jin, Qifan Song, and Jun~S Liu.
\newblock An adaptive exchange algorithm for sampling from distributions with
  intractable normalizing constants.
\newblock {\em Journal of the American Statistical Association},
  111(513):377--393, 2016.

\bibitem[MGM12]{murray2012mcmc}
Iain Murray, Zoubin Ghahramani, and David MacKay.
\newblock Mcmc for doubly-intractable distributions.
\newblock {\em arXiv preprint arXiv:1206.6848}, 2012.

\bibitem[MPRB06]{moller2006efficient}
Jesper M{\o}ller, Anthony~N Pettitt, Robert Reeves, and Kasper~K Berthelsen.
\newblock An efficient markov chain monte carlo method for distributions with
  intractable normalising constants.
\newblock {\em Biometrika}, 93(2):451--458, 2006.

\bibitem[MW96]{meng1996simulating}
Xiao-Li Meng and Wing~Hung Wong.
\newblock Simulating ratios of normalizing constants via a simple identity: a
  theoretical exploration.
\newblock {\em Statistica Sinica}, pages 831--860, 1996.

\bibitem[NFW12]{nicholls2012coupled}
Geoff~K Nicholls, Colin Fox, and Alexis~Muir Watt.
\newblock Coupled mcmc with a randomized acceptance probability.
\newblock {\em arXiv preprint arXiv:1205.6857}, 2012.

\bibitem[Pes73]{peskun1973optimum}
Peter~H Peskun.
\newblock Optimum monte-carlo sampling using markov chains.
\newblock {\em Biometrika}, 60(3):607--612, 1973.

\bibitem[PFR03]{pettitt2003efficient}
Anthony~N Pettitt, Nial Friel, and R~Reeves.
\newblock Efficient calculation of the normalizing constant of the autologistic
  and related models on the cylinder and lattice.
\newblock {\em Journal of the Royal Statistical Society: Series B (Statistical
  Methodology)}, 65(1):235--246, 2003.

\bibitem[PW96]{propp1996exact}
James~Gary Propp and David~Bruce Wilson.
\newblock Exact sampling with coupled markov chains and applications to
  statistical mechanics.
\newblock {\em Random Structures \& Algorithms}, 9(1-2):223--252, 1996.

\bibitem[Rov02]{roverato2002hyper}
Alberto Roverato.
\newblock Hyper inverse wishart distribution for non-decomposable graphs and
  its application to bayesian inference for gaussian graphical models.
\newblock {\em Scandinavian Journal of Statistics}, 29(3):391--411, 2002.

\bibitem[Wol89]{wolff1989collective}
Ulli Wolff.
\newblock Collective monte carlo updating for spin systems.
\newblock {\em Physical Review Letters}, 62(4):361, 1989.

\end{thebibliography}
\end{document}